\documentclass[acmtog,nonacm,anonymous=false,bookmarksnumbered,unicode,hidelinks,breaklinks=true]{acmart}
\usepackage{tikz}   
\usepackage[utf8]{inputenc}
\usepackage{ amsthm, amsmath, amsfonts} 
\usepackage[linesnumbered,ruled,boxed]{algorithm2e}
\usepackage{textcomp}
\usepackage{xcolor}
\usepackage[normalem]{ulem}
\usepackage{graphicx}
\usepackage{booktabs}
\usepackage{csquotes}
\usepackage{svg}
\usepackage{wrapfig}
\usepackage{subcaption}

\usepackage{xurl}
\newtheorem{thm}{Theorem}

\newtheorem{lemma}{Lemma}

\usepackage[disable]{todonotes}
\usepackage{varwidth}

\setcopyright{acmlicensed}
\copyrightyear{2026}
\acmYear{2026}
\acmDOI{XXXXXXX.XXXXXXX}

\acmJournal{TOG}
\acmVolume{0}
\acmNumber{0}
\acmArticle{0}
\acmMonth{10}

\DeclareSymbolFont{parenthesis}{OT1}{ntxr}{m}{n}

\DeclareFontFamily{U}{bbold}{}
\DeclareFontShape{U}{bbold}{m}{n}{
 <-5.5> s*[0.96] bbold5
 <5.5-6.5> s*[0.96] bbold6
 <6.5-7.5> s*[0.96] bbold7
 <7.5-8.5> s*[0.96] bbold8
 <8.5-9.5> s*[0.96] bbold9
 <9.5-11> s*[0.96] bbold10
 <11-15> s*[0.96] bbold12
 <15-> s*[0.96] bbold17
}{}
\RenewDocumentCommand{\mathbb}{m}{%
  \text{\usefont{U}{bbold}{m}{n}#1}%
}

\usepackage[belowskip=-3.5pt,aboveskip=2pt]{caption}
\def\BibTeX{{\rm B\kern-.05em{\sc i\kern-.025em b}\kern-.08em
    T\kern-.1667em\lower.7ex\hbox{E}\kern-.125emX}}

\begin{document}
\hypersetup{hidelinks}

\title{SAMS: A Structured Graph Framework for Securing Additive Manufacturing Systems}

\author{Sarah Roscoe}
\email{srr6@stmarys-ca.edu}
\orcid{0000-0003-0372-7953}
\author{Lachlan Hovany}
\email{lh39@stmarys-ca.edu}
\author{Kamila Maciocha}
\email{km82@stmarys-ca.edu}
\affiliation{%
  \institution{Saint Mary's College of California}
  \city{Moraga}
  \state{California}
  \country{USA}
}
\author{Witawas Srisa-an}
\email{witawas@unl.edu}
\author{Jitender S. Deogun}
\email{deogun@unl.edu}
\affiliation{%
  \institution{University of Nebraska-Lincoln}
  \city{Lincoln}
  \state{Nebraska}
  \country{USA}
}

\renewcommand{\shortauthors}{Roscoe et al.}

\begin{abstract}
Additive manufacturing, also known as 3D printing, is significant for its use in industries such as aerospace, robotics, medicine, and construction. However, digital mesh files remain vulnerable to geometry defects which can be introduced through malicious tampering or designer error. This can lead to hampered structural integrity and impede its use due to unexpected degradation of the part. 
To address these challenges, we propose SAMS, a multi-layered hierarchical graph framework designed to analyze meshes to enhance the Security of Additive Manufacturing Systems. SAMS captures geometric and topological information across mesh, facet, and surface levels, providing an adaptable architecture for targeted mesh processing tasks. 
We evaluate SAMS's accuracy against existing mesh processing tools. 
The graph construction results in rich structural representation and has the potential to significantly reduce downstream analytical complexity when detecting vulnerabilities. 
This paper presents SAMS, its formal graph-theoretic properties, and releases an open-source Python 3.12 implementation on GitHub (\url{https://anonymous.4open.science/r/SAMS-graph-classes/}) to facilitate future research in secure additive manufacturing systems software.
\end{abstract}

\ccsdesc[500]{Computing methodologies~Mesh models}
\ccsdesc[500]{Computing methodologies~Mesh geometry models}
\ccsdesc[300]{Computing methodologies~Shape analysis}
\ccsdesc[500]{Mathematics of computing~Graphs and surfaces}
\ccsdesc[300]{Mathematics of computing~Graph algorithms}

\keywords{mesh analysis, additive manufacturing, anomaly detection, STL file, structured graph}

\maketitle
\section{Introduction}

3D triangle meshes, particularly STL files, represent the foundational geometric structure for computer graphics and additive manufacturing (AM). AM, also called 3D printing, is a rapidly expanding industry projected to exceed \$37 billion in market size, while becoming deeply integrated into mission-critical applications across aerospace, automotive, medical, and construction industries \cite{marketsize2026,abdulhameed2019,forcael2025}. As additively manufactured parts become ubiquitous, ensuring the dependability and security of the underlying digital models is paramount.

Our aim is to provide a system for assuring the integrity of a digital model, by developing a class of graphs that facilitates extraction and analysis of mesh features that an AM user can reference in order to determine if their part is ready for printing and will function as expected. 

We pursue this line of research because detecting defects during the manufacturing process of a part can be difficult. Once the mesh is designed in CAD, many end users choose to rely on slicer software, which primarily performs computationally simple repairs such as flipped or degenerate facets, rather than evaluating underlying structural or security vulnerabilities. 

To make more substantial changes, either the user must have CAD re-construct the parametrized object, or the mesh itself must be modified by other means.
Since the mesh consists of triangles, it can be read and modified by many open-source and freely available mesh processing software. 
Such an opportunity for users also creates opportunities for malicious actors. These applications often run on networked systems, so an attacker could access and compromise a device and modify the part geometry unbeknownst to the user \cite{belikovetsky2017}. 
For example, adding a hollow space called a \emph{void} to a part can be as simple as creating a 4-sided triangular prism, adding little to the file size \cite{sturm2017}.
Other defects can be injected, such as the mutation or movement of the model's vertices.
These defects can substantially compromise the part's usefulness, structural integrity, and the security of the system in which it is used. 

Since additively manufactured parts are used in industries critical to everyday life, the presence of small modifications (such as void or the shifting of vertices, \cite{sturm2017,roscoe2025,belikovetsky2017}) in a design presents a security risk. 
Such issues give rise to the need for detection, remediation, and prevention if possible, of such modifications. 
These modifications could occur maliciously or simply through inadvertent, unintended design errors.
Manual inspection with calipers and other specification testing of a part after it has been manufactured can be accomplished. If the part fails testing, it will have to be re-produced. Any material used will have been wasted, as well as time, energy spent to manufacture the errant part, and labor costs to adjust and refine the part if necessary.

To assure security and prevent wastage of resources, current practice may employ authentication measures which can detect modifications and intervene before producing a defective part. One such measure is side-channel monitoring, such as audio recording or monitoring printer voltage of the part-in-progress being manufactured \cite{gatlin2019,rambabu2021}. Another is producing a digital twin to verify at key manufacturing stages that the part-in-progress is aligned with its design specifications \cite{zhang2020}. These measures may work well for high-scale and high-security environments, but such measures may be not be cost-effective for smaller organizations and individuals who still require non-tampered, reliable parts. 

A part's geometry describes its eventual physical form. Therefore, a primary benefit and risk posed to a part is in its geometry. A poorly designed mesh can fail readily as a printed part whether or not the manufacturing settings are carefully controlled. Conversely, a well-engineered mesh printed with proper settings may still harbor small modifications that were already present. The mesh can appear trustworthy but, when put in use, it may fail sooner than an identical, non-modified part \cite{belikovetsky2017}. Thus, it is reasonable to examine the underlying geometry to detect any such modifications. By intervening at the mesh stage, the waste of time and material which would have been spent in later stages can be avoided. One such geometry-focused, deterministic method has already been accomplished to detect voids \cite{roscoe2025}. 

Graphs have been used to represent the structure of an object \cite{dawson-haggerty2022}. Such formalisms are to some extent used in mesh processing software or research studies for purposes such as shape classification or segmentation \cite{tang2021,vecchio2023}. 
Moreover, there are taxonomies developed which use graph theory to assess risk modeling \cite{rahman2024}.
We present a unified framework that represents a mesh as a set of hierarchical graphs to analyze the geometry of a mesh. These graphs can be extended or adapted to suit a user's specific analysis needs. Our specific contributions are now described.

\subsection{Contribution}
In this paper, we use information inherent in previous graph representations of meshes and expand the representation to formalize a class of hierarchical graphs. 

We propose SAMS, a new class of structured graphs to improve the Security of Additive Manufacturing Systems. We investigate SAMS' theoretical properties and describe its applications to AM security. SAMS graphs are structured, providing levels of granularity at which mesh properties may be inspected. Specifically, our contributions are as follows:
\begin{itemize}
\item We define and formalize SAMS graphs across three granularities: mesh, facet, and surface. This contribution enables topological feature extraction and provides a foundation for defect analysis.
\item We implement SAMS software open-source and object-oriented using Python 3.12. It is available on GitHub (\url{https://anonymous.4open.science/r/SAMS-graph-classes/}). We evaluate its accuracy, file processing times, and computational throughput against established mesh processing libraries.
\item We demonstrate how SAMS may be used for security, particularly detecting localized vulnerabilities which can compromise structural integrity.
\end{itemize}

\subsection{Paper Structure}
The remainder of this paper is organized as follows. 
Section \ref{sec:background} contains a background of additive manufactuing. SAMS graphs and their properties are presented in Section \ref{sec:approach}. Implementation details, experimental evaluation, and results are described and exhibited in Section \ref{sec:implementation-eval-results}. The evaluation results and our insights on the framework as a whole are discussed in Section \ref{sec:discussion}. Several utility scenarios of SAMS graphs are identified in Section \ref{sec:usecases}. We identify related work in Section \ref{sec:related}. Finally, we conclude in Section \ref{sec:conclusion}.

\section{Background of Additive Manufacturing}\label{sec:background}

Our framework is rooted in its applications to additive manufacturing. 
In this section, we outline the background for the additive manufacturing process.

\subsection{STL Triangle Meshes}\label{sec:STL}
The most common triangle mesh is called STL. This file type, thought to stand for STereoLithography, or perhaps Standard Tessellation, or Triangle, Language, is widely used in the additive manufacturing industry to distribute 3D models. Each file contains triangles, the aim of which is typically to collectively describe the boundary of a 3-dimensional model. 
In essence (using definitions from \cite{szilvsi-nagy2003} and \cite{hamann1994}), an STL file is simply a list of one or more defined solids. Each solid has a list of triangles, each defined by 3 points in 3D space, as well as a normal vector indicating which way the triangle faces.

We now present mesh terminology.
A \emph{point}, or \emph{vertex}, is an $(x,y,z)$ point in 3-dimensional Cartesian space. An \emph{edge} is a line segment joining two points. A \emph{triangle}, also called \emph{facet} or \emph{face}, is the triangle formed by joining three distinct vertices. The \emph{normal vector}, also called simply the \emph{normal} or \emph{vector}, is the vector which is perpendicular to the facet. Since the file defines an externally-facing model, a decision must be made of which side is``outside" and which is ``inside". The normal vector specifies the correct option.

A \emph{mesh} is the collection of facets which define the boundary of an object.
An STL file describes one mesh which may have multiple independent structures that define their own boundary region. Each structure is referred to as a \emph{shell}.

Each shell of the mesh should ideally be watertight; that is, it should completely enclose a portion of 3D space. To do so, each edge in the mesh must be shared by at most two triangles, or else the mesh is called \emph{non-watertight} or \emph{non-manifold}. \\
However, throughout the design or modeling process, low-level errors can occur in the exported STL file. Sometimes faces are flipped, or the model has self-intersections. These computational errors can usually be repaired by slicers, but more complex errors cannot.

Errors in the STL file are a relevant issue in the security of additive manufacturing. With its digital components, the cybersecurity of the manufacturing process must be ensured, beginning with determining that the model file is free from tampering. The SAMS graph framework has been developed to continue efforts in this regard. For more information on the importance of mesh and additive manufacturing security, see Section \ref{sec:threatmodel}.

\subsection{Additive Manufacturing Workflow}
\label{sec:AM}
Our objective in this project to secure the digital aspects of the manufacturing process. We focus on the mesh itself for several reasons. We do not choose to analyze in a particular computer-aided design (CAD) software, because while the myriad of applications available share the same basic objectives, their underlying design, cost structures, and ease of use can vary substantially. However, CAD software supports exporting to mesh format, which is done to distribute or sell meshes without compromising intellectual property \cite{osborn2017}, leading to a large availability of files on the internet.

Additive manufacturing (AM) begins with a digital model created by a CAD program such as Blender or AutoCAD, which is then exported to a geometry file called a \emph{mesh}, often in STL format.
A \emph{slicer} such as PrusaSlicer or Cura then slices the mesh into layers and stores the instructions to construct the part in GCODE format, which is executed on the manufacturing machine by laying material.
Post-processing includes any physical measures that needs to be completed before the part is put in use, which may include washing and sintering metal parts, or sanding and painting non-metal parts.

Once the mesh is generated, anyone can download a mesh file and run it through their slicer software of choice. This software may have the same limitations as CAD (variations of implementation, cost, and ease of use). Examining the mesh itself is an efficient and cost-effective entry point to determining the viability of printing a mesh, limited only to the cost of obtaining the mesh files. Slicers are typically equipped with some ability to repair broken meshes, such as flipped faces, but the ability can vary widely between applications. Mesh analysis that goes beyond simple repair is not standardized, and can be restricted to software that is not openly available. 

For these reasons, we focus on mesh analysis, after it is exported from CAD and before it is sent to the slicer.

With a description of additive manufacturing and meshes complete, we can proceed to how meshes may be represented as a graph framework.

\section{SAMS graphs: A Graph-Theoretic Approach}\label{sec:approach}

The following section presents an approach to the representation of STL digital model files as graphs. It begins with a review of graph terminology in Section \ref{sec:graph-terminology}, and an overview of structured graphs in \ref{sec:structuredgraphs}. The SAMS framework is then defined as a structured graph in Section \ref{sec:definitions}. Properties of the graphs are described in \ref{sec:graphproperties}. A possible reduction of the first two levels of the framework, and an algorithm for obtaining it, is presented in \ref{sec:surfacegraph}. An overview of SAMS's Python implementation is given in Section \ref{sec:implementation}.

\subsection{Review of Graph Terminology}\label{sec:graph-terminology}

We adopt standard formal graph-theoretic notation which can be found in any graph theory book, like \cite{chartrand2010}. A \emph{graph} is defined as an ordered pair $G = (V, E)$, where $V$ is the set of \emph{vertices} (or nodes) representing discrete entities, and $E \subseteq \{\{u, v\} \mid u, v \in V, u \neq v\}$ is the set of \emph{edges} representing unordered pairs of relationships between them.

Strictly defined, graph vertices are topological entities and do not possess inherent coordinates in Euclidean space. However, a \emph{graph embedding} maps vertices to positions in $\mathbb{R}^n$, $n\geq 2,$ allowing geometric attributes such as spatial coordinates to be associated with elements of $V$ or $E$ as required.

A graph $H = (V', E')$ is a \emph{subgraph} of $G = (V, E)$, denoted $H \subseteq G$, if $V' \subseteq V$ and $E' \subseteq E$ such that for every edge $e = \{u, v\} \in E'$, the incident vertices $u$ and $v$ are contained in $V'$.

A graph is \emph{connected} if a path exists between every pair of distinct vertices. Connectivity is further quantified by the graph's robustness: a graph is \emph{$k$-vertex-connected} (or \emph{$k$-edge-connected}) if it remains connected following the removal of fewer than $k$ vertices (or edges).

The \emph{degree} of a vertex $v$, denoted $\deg(v)$, is the number of edges incident to it. A graph is \emph{$k$-regular} if $\deg(v) = k$ for all $v \in V$; specifically, a 3-regular graph is referred to as a \emph{cubic} graph.

A graph is \emph{planar} if it admits an embedding in the plane $\mathbb{R}^2$ such that no two edges intersect except at their common endpoints. By Kuratowski’s Theorem \cite{tutte1963}, a graph is non-planar if and only if it contains a subgraph that is a subdivision of the complete graph $K_5$ or the complete bipartite graph $K_{3,3}$.

A planar embedding partitions the plane into regions called \emph{faces}. The \emph{dual} $G^*$ of a planar graph $G$ is constructed by associating a vertex in $G^*$ with each face of $G$ and connecting two vertices in $G^*$ if their corresponding faces share an edge in $G$ (Figure \ref{fig:graph-dual}).

\begin{figure}

	\centering

	\includegraphics[width=0.5\columnwidth]{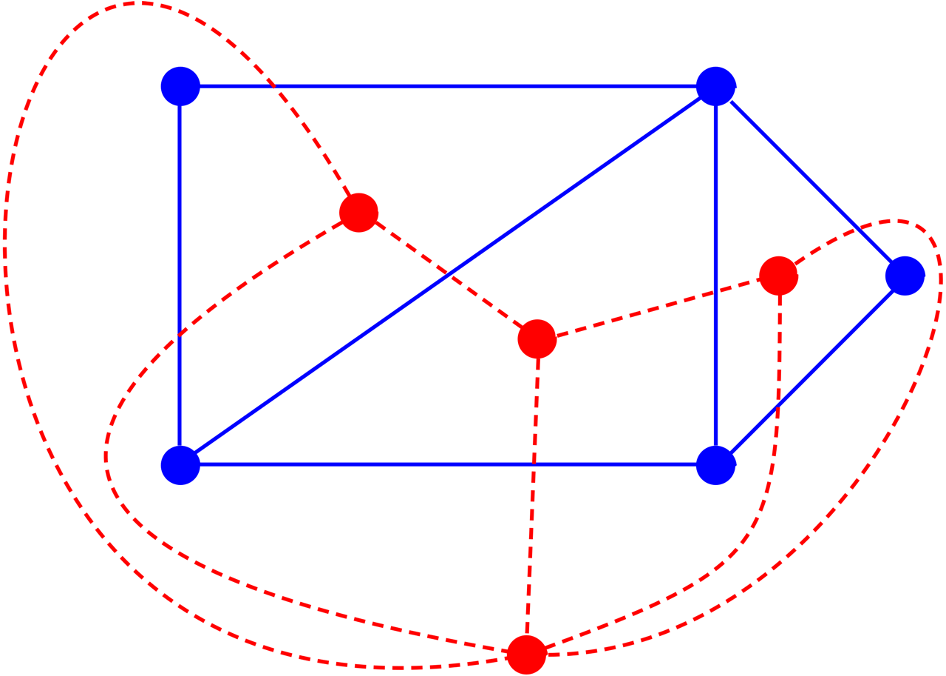}

	\caption[Graph dual.]{A graph (blue, solid lines) and its dual (red, dotted lines). Image is open-source from \cite{dualgraph}.}

	\label{fig:graph-dual}

\end{figure}

 In the context of mesh processing, the dual is constructed by collapsing each facet into a single graph vertex, though we do not include the vertex that corresponds to the space outside the mesh. We refer to this dual graph as the \emph{facet graph} in our work.

\begin{figure}[ht]
    \centering
    \includegraphics[width=0.5\columnwidth]{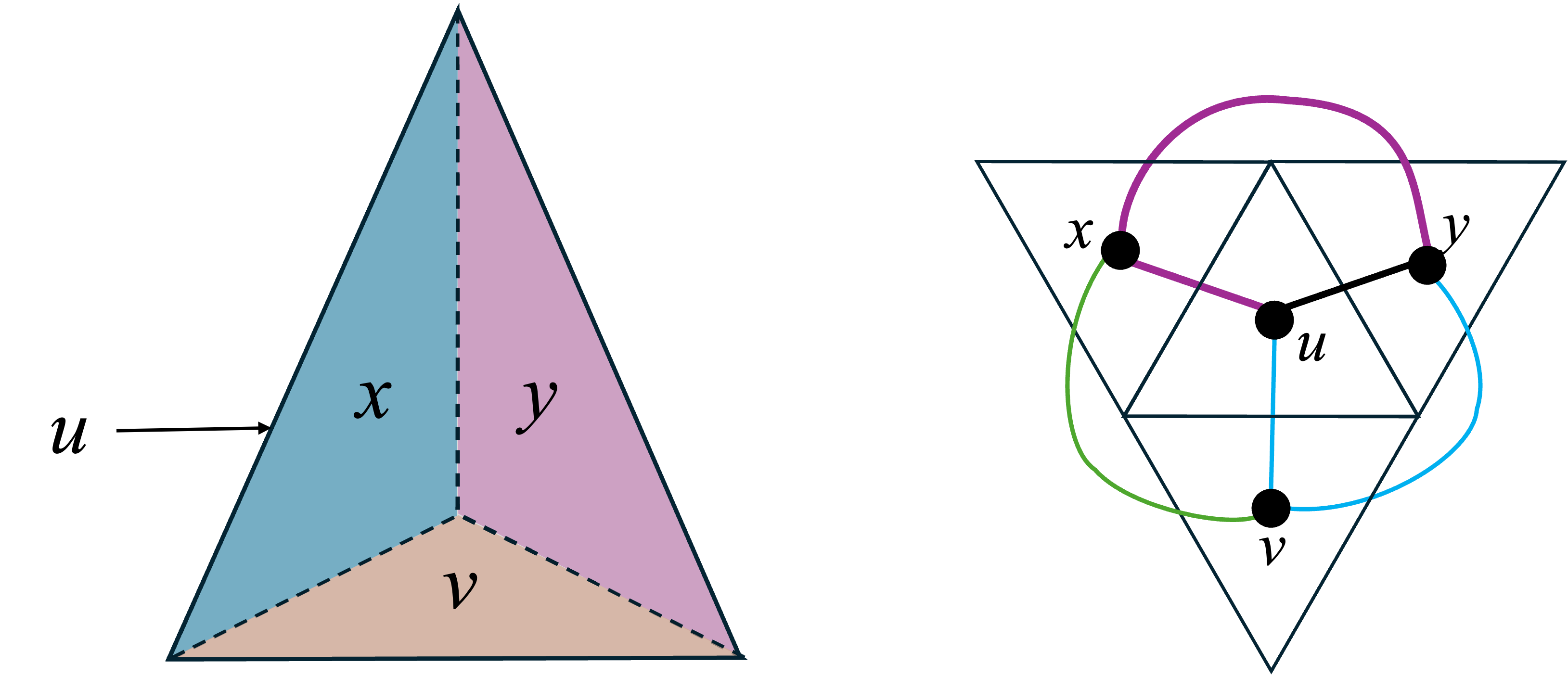}
    \caption[A tetrahedron.]{A tetrahedron and its facet graph. $u$ is the front-facing facet, $v$ is the bottom facet, while $x$ and $y$ are on the sides.}%
    \label{fig:3ctd-base}
\end{figure}

\subsection{Overview of Structured Graphs}\label{sec:structuredgraphs}

A structured graph is a graph which allows for grouping of some connected vertices into a combined vertex. These graphs are hierarchical and allow varying levels of abstraction. Structured graphs were first defined in \cite{defloriani1983,ancona1986,defloriani1989}, and we use the standard terminology and descriptions here.

Formally, using definitions established in \cite{defloriani1989}, a structured graph involves \emph{reducible components}, each of which are $G'=(V',E')$, connected subgraphs of the original graph $G$. Each vertex $v \in V'$ which is adjacent to a vertex $u\in V-V'$ is a \emph{boundary vertex}. The component $G'$ can be reduced to a single \emph{macrovertex} $x$ not already in $V$, thus replacing $G'$ in the graph. When $x$ is present in the graph, we call it the \emph{reduced graph} $G^R$. If $x$ is expanded, $G'$ is the \emph{expansion graph} of $x$. See Figure \ref{fig:reducing} for a visualization.

\begin{figure}[h]
    \centering
    \includegraphics[width=\columnwidth]{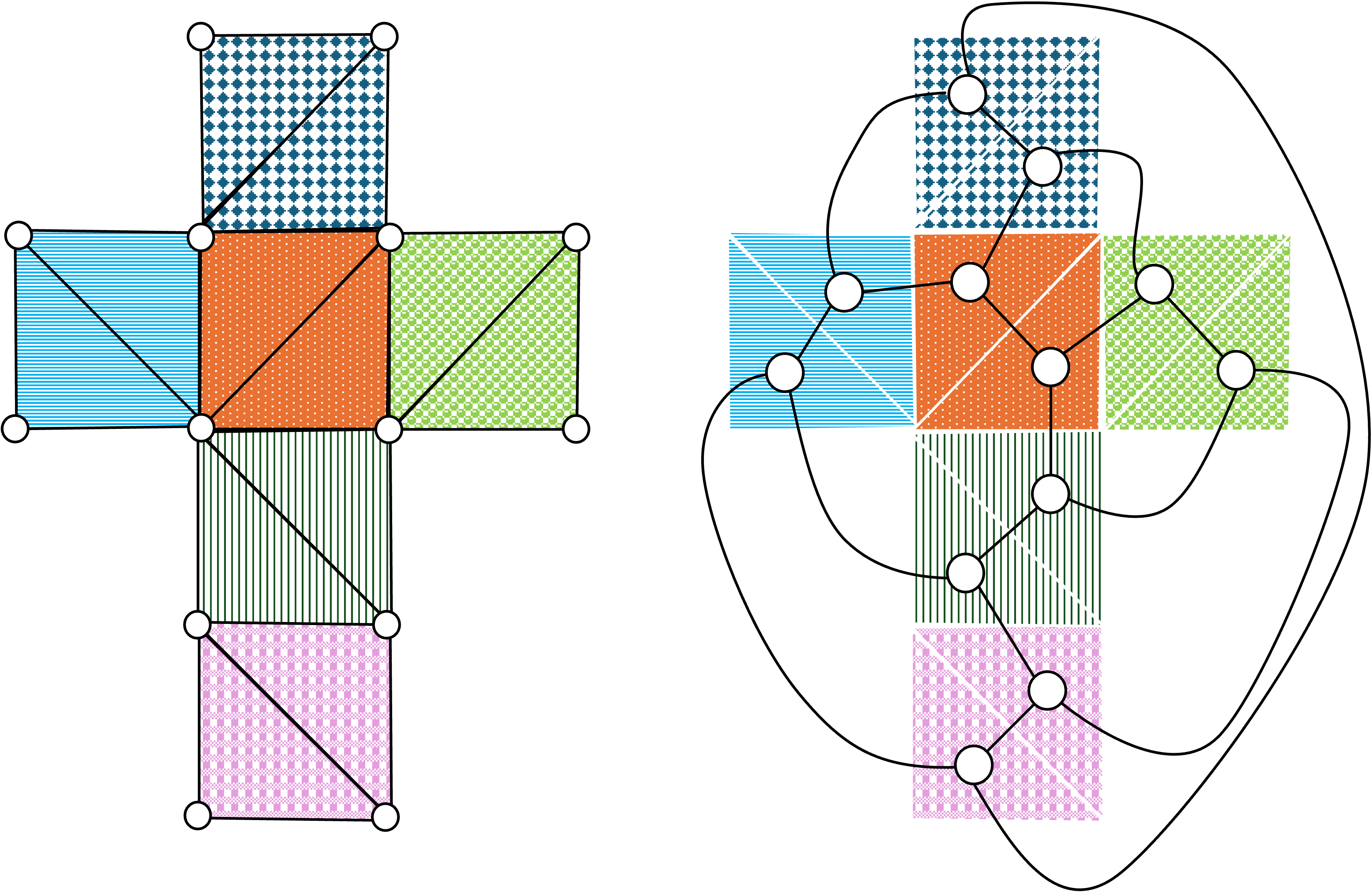}
    \caption{A depiction of a cube that has been laid flat. Each group of three vertices has been reduced to a macrovertex representing a single triangle. Since each macrovertex is adjacent to another triangle, every vertex point in the original mesh is a boundary vertex.
    }
    \label{fig:reducing}
\end{figure}

A helpful analogy to understand structured graphs is digital maps. When zoomed out on an area, only large cities are visible and named, and only major highways are shown. Zooming in reveals smaller cities and towns, as well as arterial roads. When a map is fully zoomed in, side streets are named and local landmarks are visible. This collapsing-expanding structure is a hallmark of structured graphs \cite{defloriani1983,ancona1986,defloriani1989,ancona1990,ancona1991,cai2005} and we employ it in SAMS graphs.
In the following section, properties of SAMS graphs will be described.

\subsection{From Structured Graphs to SAMS}\label{sec:definitions}
\subsubsection{Summary}
SAMS graphs consist of multiple levels of abstraction on a mesh. The levels include the mesh itself: the low-level representation of the mesh's vertices, edges, and triangles. We call the graph formed by the points and edges of the mesh the \emph{mesh graph}. The next level is the dual of the mesh graph, which we call the \emph{facet graph}. Each triangle in the mesh graph is reduced to a vertex. Edges of the facet graph are defined by triangle adjacency. The vertices of the facet graph may be clustered according to specified criteria of interest to the user. For example, the dihedral angle between facets can be computed from the triangles' normal vectors, and each angle can be classified using a threshold to determine which groups of triangles constitute surfaces in the mesh. Once the clustering is completed, the selected vertices are combined into a single vertex, and edges are drawn according to the boundary facets' edges to other clusters. We call the graph that results from this clustering the \emph{surface graph}. These three graphs describe the structure of the SAMS graph for a particular model. 

\subsubsection{Definitions}

Now the structure of SAMS is described formally. We present a framework with three levels: the mesh graph, facet graph, and surface graph. Figure \ref{fig:AMMA-graph} shows the current three steps of the framework.

\begin{figure}[h]
    \centering
    \includegraphics[width=\columnwidth]{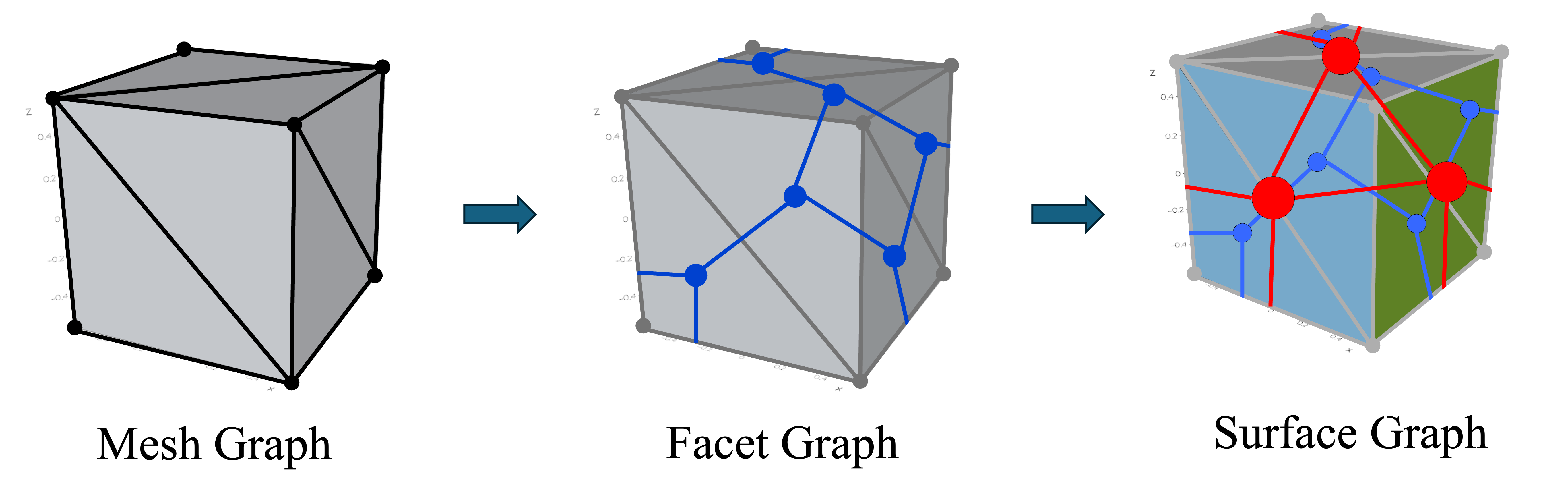}
    \caption{The three levels of the SAMS structured graph. First, the mesh graph, which is the raw mesh input, is constructed with vertices and edges in 3D space. Second, the facet graph, which is the dual of the mesh graph. Third, the surface graph reduces facets into clusters (the clustering is shown here by colors) into a single vertex. The reduced vertices are then joined by triangle adjacency.}
    \label{fig:AMMA-graph}
\end{figure}

To construct the \emph{mesh graph}, the vertex coordinates of the mesh are parsed from the STL file, three vertices at a time. Each vertex's coordinates is added to a set of graph vertices. Vertices are retained if their coordinates do not match a vertex already present in the set. Edges are inferred to exist between each of the three vertices.  In this way, the mesh graph $G_M = (V_M,E_M)$ is constructed as $$V_M = \{i, (x,y,z) | i \in\mathbb{Z}, (x,y,z)\in \mathbb{R}^3\}$$ and $$E_M = \{(u,v) | u, v \text{ are two vertices in the same facet}\}.$$

Vertices and edges alike are each assigned an index according to their order of appearance in the STL file. From this point forward, we avoid confusion and refer to vertices and edges derived from the mesh as \emph{mesh vertices} and \emph{mesh edges}. 

The natural structure of an STL file provides a simple transition to the next category of graph. Each facet is treated as a vertex in the \emph{facet graph}, defined as $G_F = (V_F,E_F)$ where the set of vertices is formed by the set of facets. In such a way, each facet is a macrovertex for each of the three mesh vertices; the subgraph formed by the 3-clique of mesh vertices has been reduced to a facet vertex. Formally, $$V_F = \{ v=(a,b,c) | a, b, c\in V_M \text{ and form a 3-clique} \}$$ and $E_F$ is defined by the boundary edges of the cluster: if a facet vertex $v$ is the macrovertex of mesh vertices $(a,b,c)$, and both $a$ and $b$ are adjacent to a third mesh vertex $d$, there is an edge between facet $v$ and the facet $u$ formed by vertices $(a, b, d)$. Therefore, \begin{align*}
E_F = \{(u,v) &| u=(a,b,c) \text{ and } v=(a,b,d), 
\\ &\text{ where } a,b,c,d\in V_M; ab, bc, ac, ad, bd \in E_M\}.
\end{align*}

Once the facet graph is constructed, facets may be clustered according to many possible heuristics. One such heuristic includes grouping triangles by similar area, or by the angle between any adjacent pair of facets. The following level of abstraction is modular, and can result in many possible graphs. We consider one type of clustering, grouping triangles based on similar angles. In such a way, we define the \emph{surface graph}, $G_{SURF} = (V_{SURF},E_{SURF})$, where $V_{SURF}$ is a macrovertex of selected facet graph vertices and is called a \textit{surface}. Each $v\in V_{SURF}$ is non-overlapping, meaning $v$ cannot be part of two surfaces at once, nor can a facet vertex be part of more than one surface. $E_{SURF}$ is constructed by the boundary edges between surfaces: If facet $f$ is contained in vertex $v\in V_{SURF}$, and facet $g$ is contained in vertex $u\in V_{SURF}$, and $f$ and $g$ are adjacent in the facet graph, then there exists an edge in the surface graph joining $(u,v)$. This is one possible structured level of SAMS graphs; the framework as proposed is modular and can be adapted for many different purposes. We detail the clustering algorithm in Section \ref{sec:surfacegraph}.

We now describe some interesting and important graph-theoretic properties of SAMS graphs.

\subsection{Graph-Theoretic Properties of SAMS Graphs}\label{sec:graphproperties}

With the SAMS framework defined, we now address various properties the structured graph can possess. For the sake of simplicity, we will proceed in this section assuming a 2-manifold mesh that has no open edges or singleton vertices.

\subsubsection{Shells}

A shell is a maximal group of connected facets in a mesh. Having more than one shell in a mesh file allows for multiple independent pieces to be printed at once, or can help define internal complexity, such as defining a hollow space enclosed in a part. This geometric concept corresponds directly to connected components in both the mesh and facet graphs. Since each shell is a distinct structure which is independent from other shells by a nonzero distance, there will not be a path from mesh vertices or facets of one particular shell to those vertices or facets of another shell. This allows for shells to be examined independently, each as its own component of the overall part.

\subsubsection{Degree and Connectivity}
Now we investigate the degrees of facet graphs.

\begin{thm}\label{thm:cubic}
    Every well-defined facet graph is a cubic graph.
\end{thm}
\begin{proof}
    We assume the facet graph is obtained from a 2-manifold mesh. The property assures that every edge of the mesh graph is shared by exactly two unique facets. Thus, every facet in the mesh will have three unique neighbors. Since the facet graph is the dual of a mesh graph in which every face is a triangle, thus every vertex of the facet graph will have degree three. Therefore, the facet graph is cubic. 
\end{proof}

\begin{thm}\label{thm:ctcy}  
Every facet graph is at least two-connected.
\end{thm}  
\begin{proof}  
We prove this by contradiction. So every facet graph must be at most one-connected. Consider the graph at the top of Figure \ref{not1ctd}. It is a cubic graph and one-connected, but attempting to construct a mesh from the graph vertices results in an impossible structure according to our assumptions.   
Looking at facets E and D, the point of D's triangle must go either to the point containing B, C, G, H, E, F, or it goes to the point containing only E, F, and I. Choosing one point aligns D's second edge, but choosing either point results in the remaining edge being open, as D cannot have more than three points. Thus the graph cannot be a facet graph, and there is at least one facet graph that is not one-connected, a contradiction. 

\begin{figure}[h]
    \centering
    \includegraphics[width=\columnwidth]{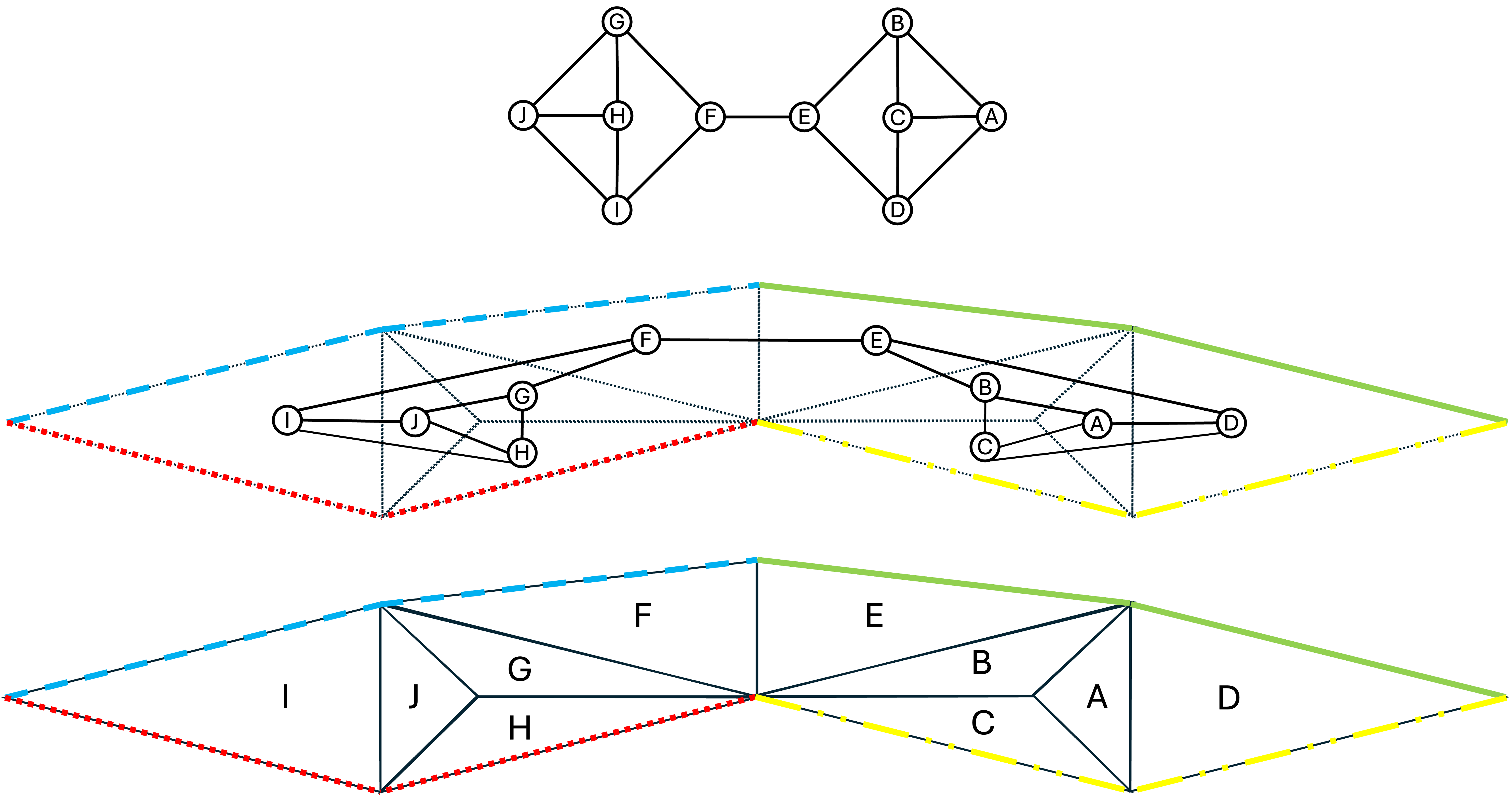}
    \caption[A one-connected cubic graph which is not a facet graph.]{Top: A one-connected cubic graph which is not a facet graph. Middle: The graph overlaid on a possible configuration of its triangles. Bottom: The configuration of triangles without the graph. Colored and dashed edges will be adjacent when the object is folded.}
    \label{not1ctd}
\end{figure}

Moreover, consider the structure of a two-manifold, three-dimen\-sional object. Each mesh vertex represents the joining of several facets. Each vertex of the mesh therefore will correspond to a cycle in the facet graph, by definition of the dual. It is clear to see that an induced subgraph in the facet graph of each corner of the mesh will be two-connected. And since the mesh is assumed to have no holes or missing triangles, every corner creates an induced cycle. Therefore the facet graph must be at least two-connected. 
\end{proof}

\subsubsection{Genus and Planarity}
For simplicity, we will proceed considering only single-shell objects.
Mesh graphs are analogous to polyhedral graphs \cite{tutte1963}, which are formed by the vertices and edges of convex, three-dimensional polyhedra. This graph is known to be planar. 
By embedding the vertices and edges of the convex hull of the mesh or the facet graph on a convex object like a sphere, it is possible to see that the mesh and facet graph of an object will be planar iff the genus of the object is 0. 
The \textit{genus} of a 3D object, intuitively, is the number of holes an object possesses. The genus stays the same if the object is mutated, molded or modified to have any shape, as long as any holes through the object are retained. Examples of objects with a variety of genera can be found in Figure \ref{fig:mesh-genus}.

\begin{figure}[h]
    \centering
    \includegraphics[width=\columnwidth]{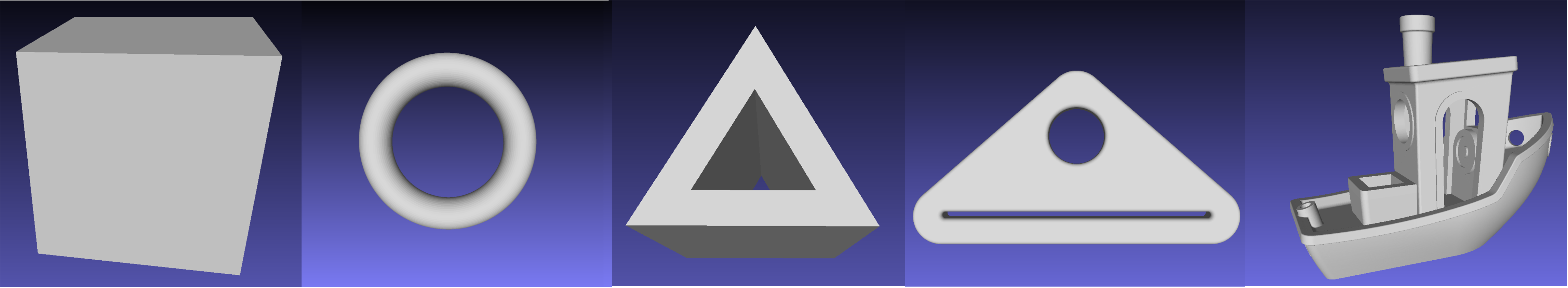}
    \caption{Objects with genus $\ge 0$. From left to right: a solid cube, genus 0; a ring, genus 1; a triangular tube, genus 1; a toothpaste tube holder, genus 2; a toy boat ``Benchy", genus 7.}
    \label{fig:mesh-genus}
\end{figure}

If the genus of the object is greater than 0, it cannot be embedded on the sphere, and must be topologically equivalent to the compact sum of tori and disks \cite{gallier2012}. One such graph that must be embedded on a torus to be planar is the utility graph, $K_{3,3},$ the complete bipartite graph on 6 vertices. This raises the question of whether the facet and mesh graphs in general are necessarily planar. 

\begin{lemma}\label{not-every-fg-planar}
    Not every facet graph is planar.
\end{lemma}

\begin{proof}
    This can be seen with an object like a triangular tube. It is a triangular prism with a cutout of a smaller triangular prism that goes through the entire length of the object. Figure \ref{fig:tritube} shows the object first as seen in a mesh-viewing program, then its facets drawn and numbered, and finally its facet graph. The graph shown is nonplanar. The resulting facet graph has a {subdivision} of $K_{3,3}$, which is the subgraph consisting of the bold edges in Figure \ref{fig:tritube-subdiv}. 
Therefore, there exists a mesh with a non-planar facet graph.
\end{proof}

\begin{figure}[h!]
    \centering
    \includegraphics[scale=0.25]{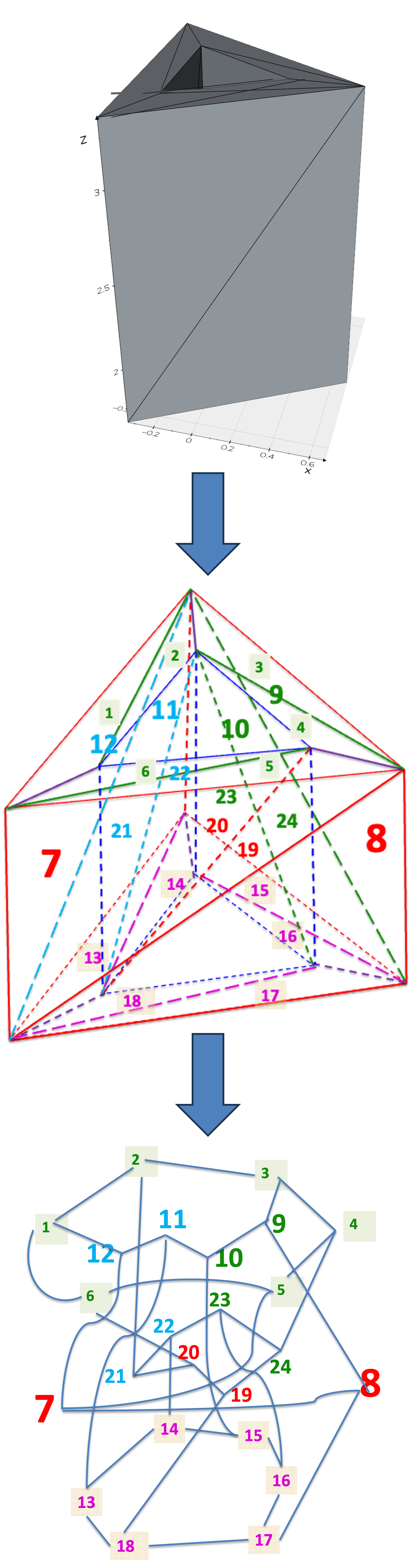}
    \caption{A triangular tube (top), its facet triangulation (middle), and its facet graph (bottom).}
    \label{fig:tritube}
\end{figure}

\begin{figure}
    \centering
    \includegraphics[width=\linewidth]{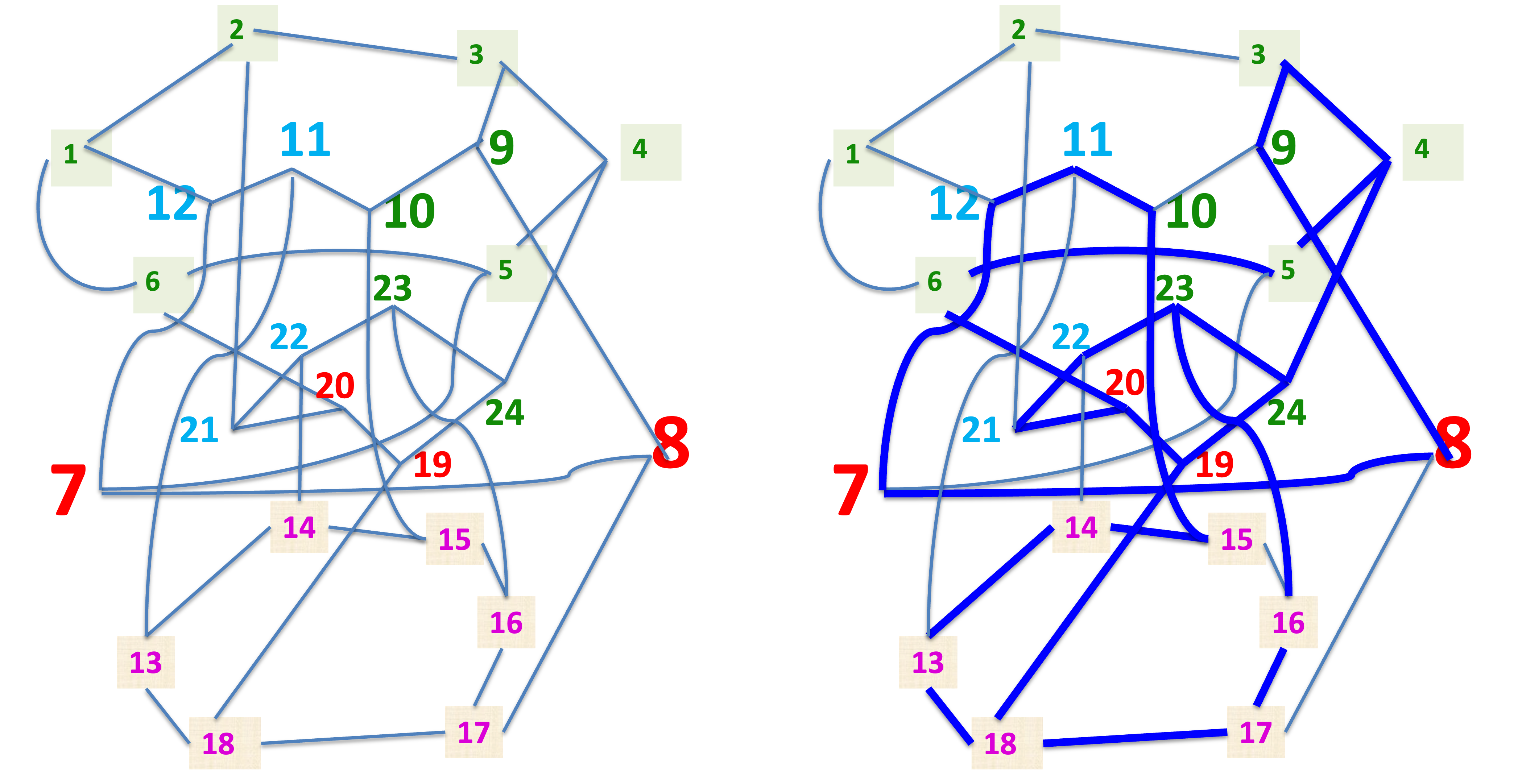}
    \caption{The original facet graph of the triangular tube (left), and its subdivision emphasized in darker edges (right).}
    \label{fig:tritube-subdiv}
\end{figure}

Kuratowski's theorem \cite{tutte1963} states that any non-planar graph has either a subdivision of $K_{3,3}$ or $K_5$. Having found an object with a $K_{3,3}$ subdivision, let us briefly examine the possibility of having a $K_{5}$ subdivision.

\begin{thm}\label{thm:k5}
    $K_5$ or any subdivision of it cannot be a facet graph, or a subgraph of a facet graph. 
\end{thm}

\begin{proof}
    $K_5$ is the complete graph on five vertices, and thus each vertex has degree four. By Theorem \ref{thm:cubic}, the facet graph is cubic. Therefore, $K_5$ cannot be a facet graph on its own or a subgraph of a facet graph. 
    Moreover, a subdivision of $K_5$ being present in a graph requires that there be 5 vertices with at least degree 4. This is an impossible construction for a facet graph for the same reason. Therefore, it is not possible that $K_5$ is a facet graph, a subgraph, or subdivision of a facet graph.
\end{proof}

The notion of a geometric dual applies to non-planar graphs with the assistance of matroid theory \cite{whitney1932}.  A non-planar graph may also be embedded on a surface such that it becomes planar \cite{diestel2025}. We have shown that it is possible for the facet graph to be non-planar. We still choose to apply the term \enquote{dual} to the facet graph, as the notion of mapping each mesh facet to a facet vertex remains relevant. 
To the best of our knowledge, a 3D mesh's planarity properties have not been previously discussed in the literature. We link the genus of an object to the planarity of its facet graph.

\subsubsection{Convexity and Concavity}\label{sec:concavity}

Mesh graphs may be considered analogous to polyhedral graphs \cite{tutte1963}, which are formed by the vertices and edges of convex, three-dimensional polyhedra. 
A critical difference is that the polyhedron from which we obtain the polyhedral graph must be convex. However, one may be transformed into the other by reducing the SAMS mesh graph to its convex hull. The SAMS mesh graph may be obtained in the other direction by adjusting vertex coordinates from the convex hull to their original positions. For this reason, we posit that SAMS graphs of genus 1 are an extension of polyhedral graphs. The polyhedral graph has been proven to be 3-vertex-connected \cite{balinski1961} which supports our conclusion in Theorem \ref{thm:ctcy}.

In the implementation of SAMS, each vertex is classified as concave, convex, planar, or other, with respect to its local positioning. To do so, first the \emph{vertex normal} is computed as the sum of all face vectors incident to a mesh vertex. To compute concavity, the mesh vertex's incident edges are converted to vectors by subtracting the coordinate of the current mesh vertex from the coordinates of adjacent mesh vertices. Then, the angle between the vertex normal and each ``edge vector" is added to a list. The list of angles is classified as either convex (all greater than $90^{\circ}$), concave (all less than $90^{\circ}$), planar (all equal to $90^{\circ}$), or other.

\subsection{Surface Graph: a Reduction of the Facet Graph}\label{sec:surfacegraph}

The facet graph is a well-utilized reduction of the mesh graph in the literature, if not referred to in structured graph language as described here. The facet graph may be further reduced by a variety of algorithms to analyze the mesh in coarser focus. This can allow for mesh segmentation algorithms which do not require computationally heavy algorithmic processes such as convolutional or graph neural networks. We propose one such algorithm in this section. We choose to reduce the facet graph by clustering one or more triangles into \emph{surfaces}. Surfaces may consist of flat, planar collections of triangles, such as the sides of a cube, or they may consist of continuous curves with small angles between the triangles as in a sphere or the sides of a vase. 
One surface is distinguished from another by magnitude of the angle between them. If the angle between two triangles is sharp, they are clustered separately. If the angle is more subtle, the clusters are joined. The sharpness of the angle is arbitrary and determined by the user.

We accomplish a clustering of the entire object using an intuitive, deterministic algorithm which is linear on the order of edges in the facet graph. As a result, the {surface} graph is obtained. The vertices of the surface graph are the clusters of triangles; an edge in the surface graph represents two distinct surfaces being contiguous. Surfaces themselves, while holes may be present, contain triangles which are all contiguous: every triangle in the surface must share a facet edge with at least other facet in the same surface, unless the surface consists of only once facet. In this way, every facet vertex must be part of a unique surface macrovertex.

To construct the graph, we employ Algorithm \ref{alg:partitionToSurface}.
The mesh and a static threshold value are the inputs to the algorithm. Initially, every facet is contained in its own cluster, so $|V_{SURF}| = |V_F|$. Iterating through the set of facet edges, we compare the angle between two facets. If the angle is below the threshold, the two clusters are joined by merging the smaller-sized cluster to the larger one. Each facet in the smaller cluster is assigned to be part of the larger partition. Alternatively, if the angle is above the threshold, no action is taken.

Once every edge has been processed, the output of the comparison is a map of cluster IDs to facet vertices. The cluster IDs are the vertices of the surface graph. An additional iteration of the edges is required to determine the edges of the surface graph: for each facet edge, consider whether the incident vertices are part of different clusters. If they are, then the facets are boundary facets of the surface, and there is an edge between these two clusters in the surface graph $G_{SURF}$. 

The user-set threshold greatly impacts the resulting surfaces. While for simple objects, such as a cube, the threshold between a new cluster is necessarily stark (90$^\circ$), but for other objects it may be more subtle.

The runtime of the algorithm is $O(E_F)$, since all edges of the facet graph will be processed. This is much faster than a traditional graph traversal such as BFS, which would be $O(V_F+E_F)$. The proposed algorithm is only concerned with angles between individual vertices, so no additional information is needed.

\begin{algorithm}
\caption{PartitionByAngles(design, threshold)}
\label{alg:partitionToSurface}
\small
\KwData{Facet graph and an angle threshold}
Create an empty dictionary mapping partitionID to the facets contained in that partition\;
Assign each facet to its own partition\;
\For{every edge $e$ in design}{
    Obtain $e$'s two facets, $f$ and $h$\;
    Determine the angle between $f$ and $h$\;
    \If{the angle between facets $f$ and $h$ is less than $threshold$}{
            Check size of $f$'s and $h$'s partitions\;
            Add both $f$ and $h$ to the larger partition if not currently present\;
            Remove $f$ and $h$ from previous partition, if it was in a different one\;
        }
    \Else{Do nothing\;}
}
\end{algorithm}

The ability to accomplish this operation deterministically has important implications. The surface graph may be used to re-triangulate each surface, as in \cite{defloriani1989}. It can also be used to help ensure geometric properties of the digital models. Section \ref{sec:usecases} provides more information about the utility of the surface graph and {\sc SAMS} as a whole.

\section{Implementation and Evaluation}\label{sec:implementation-eval-results}

The SAMS framework has been implemented and assessed for accuracy. This section describes its implementation strategy (Section \ref{sec:implementation}), evaluation process (Section \ref{sec:evaluation}), and results (Section \ref{sec:results}).

\subsection{Implementation of SAMS}\label{sec:implementation}

SAMS graphs were implemented in Python classes. While our implementation does not rely on any existing mesh processing software, it may be used to interface with visualization libraries, such as Vedo \cite{vedo}. The class levels are structured beginning with \texttt{Mesh}. From there, abstractions are possible to \texttt{FacetGraph} and \texttt{SurfaceGraph}.

In addition to coordinates and normal vectors found in an STL file, our graph classes include additional computed properties that allow for other geometric operations. These properties include the following:

\begin{itemize}
    \item \textit{Vertex normal}: sum of the normal vectors of a vertex's incident facets.
    \item \textit{Concavity}: either concave, convex, planar, or other and computed based on adjacent normal vectors, as described in Section \ref{sec:concavity}.
    \item \textit{Area}: of a facet. Stored at mesh and facet level.
    \item \textit{Centroid}: the geometric center of the triangle, computed by 
    the average of each vertex's coordinates. %
    \item \textit{Length of edge}: computed by Euclidean distance.
    \item \textit{Adjacent angles}: the angle between two normal vectors.
\end{itemize}

At each level of the graphs, references are employed between objects that are stored in the level below. As graph levels are traversed, many properties are retained across the implementation. Thus, to save memory and avoid re-creating objects, aliases are used. As a result, it is simple to refer to the reduced surface graph from the mesh level, and vice versa. This enables ease of implementation for accessing low-level details that may be needed in higher-level computations. A \texttt{Mesh} has three associated classes: \texttt{MeshVertex}, \texttt{MeshEdge}, and \texttt{MeshFace}. Each type of graph, \texttt{FacetGraph} and \texttt{SurfaceGraph}, has three associated classes: \texttt{Vertex}, \texttt{Edge}, and \texttt{Graph}. 

\subsubsection{Implementation of Mesh Graph}

\begin{figure}[h]
    \centering
    \includegraphics[width=\columnwidth]{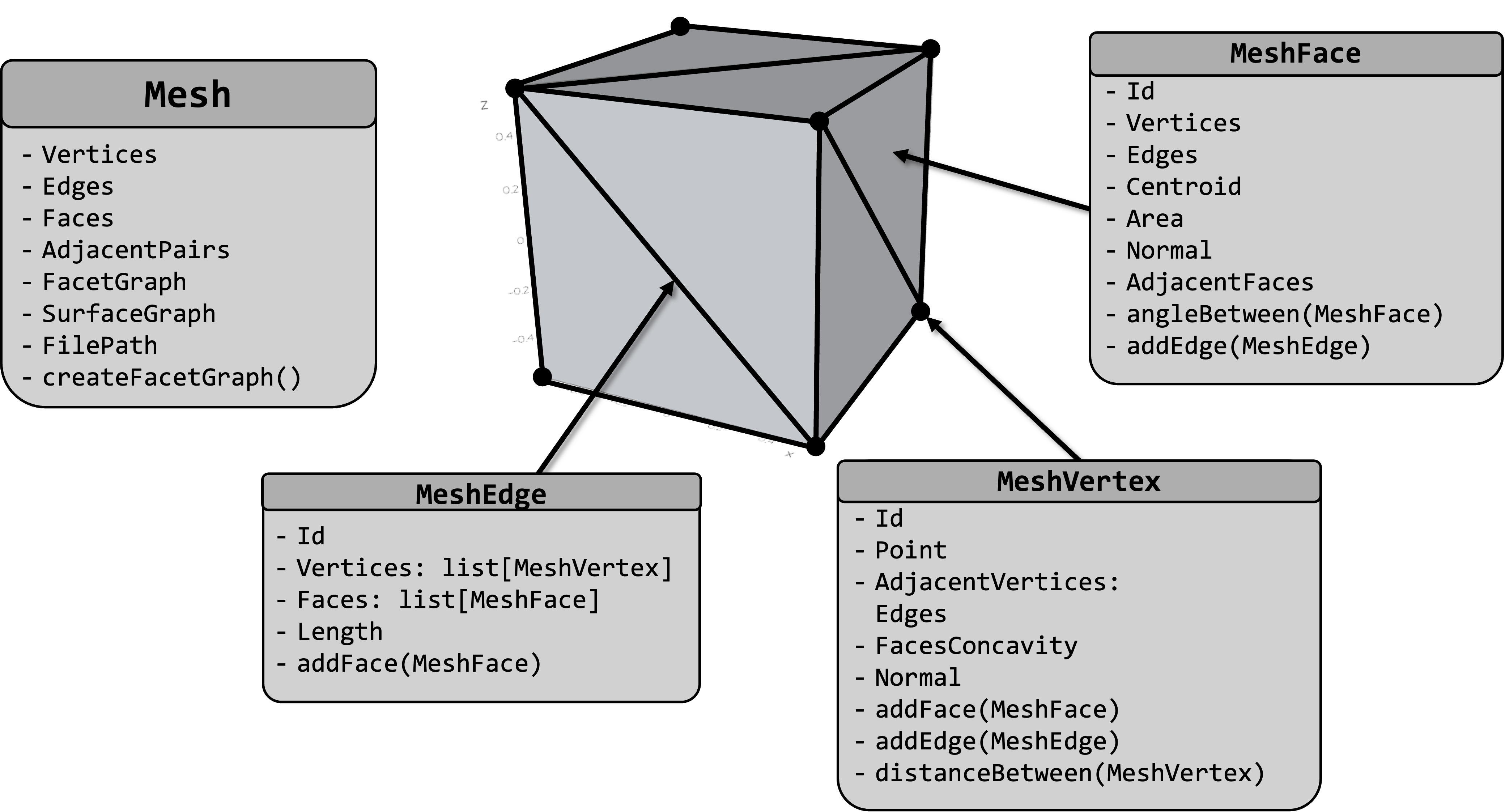}
    \caption{Description of data structures for our \texttt{Mesh} classes.}
    \label{fig:meshUML}
\end{figure}

We have implemented a mesh class (Figure \ref{fig:meshUML}) which parses and reads the triangle data within an STL mesh file and then constructs a \texttt{Mesh} object. From here, we can construct our related classes. 

\texttt{MeshFace} contains information about a face of the mesh. This information includes: a list of pointers of \texttt{MeshVertex} and \texttt{MeshEdge} objects the \texttt{MeshFace} is a part of. It also includes geometric information such as the centroid, total area, and the face in question's normal vector.

The \texttt{MeshVertex} object contains a list of references to the objects of \texttt{MeshFace} and \texttt{MeshEdge} it is part of, and adjacent to. The class also contains geometric properties including the 3D point coordinates, a vector called the vertex normal, and a measure of the vertex's concavity.

\texttt{MeshEdge} represents the line segment connecting two \texttt{MeshVertex} objects. It has the properties length (Euclidean distance between two points), the list of \texttt{MeshVertex} and \texttt{MeshFace} it joins, as well as the angle between the faces. When parsing the file, \texttt{MeshEdges} are computed based on vertices in a face. Based on this, there is no standard condition that an edge should have more or less than two \texttt{MeshVertex} objects.

Finally, the \texttt{Mesh} class holds lists of \texttt{MeshVertex}, \texttt{MeshEdge}, and \texttt{MeshFace} objects created from the STL file. It also contains references to \texttt{FacetGraph}, and \texttt{SurfaceGraph} objects that are built upon the mesh. 

\subsubsection{Implementation of Facet Graph}

\begin{figure}[h]
    \centering
    \includegraphics[width=\columnwidth]{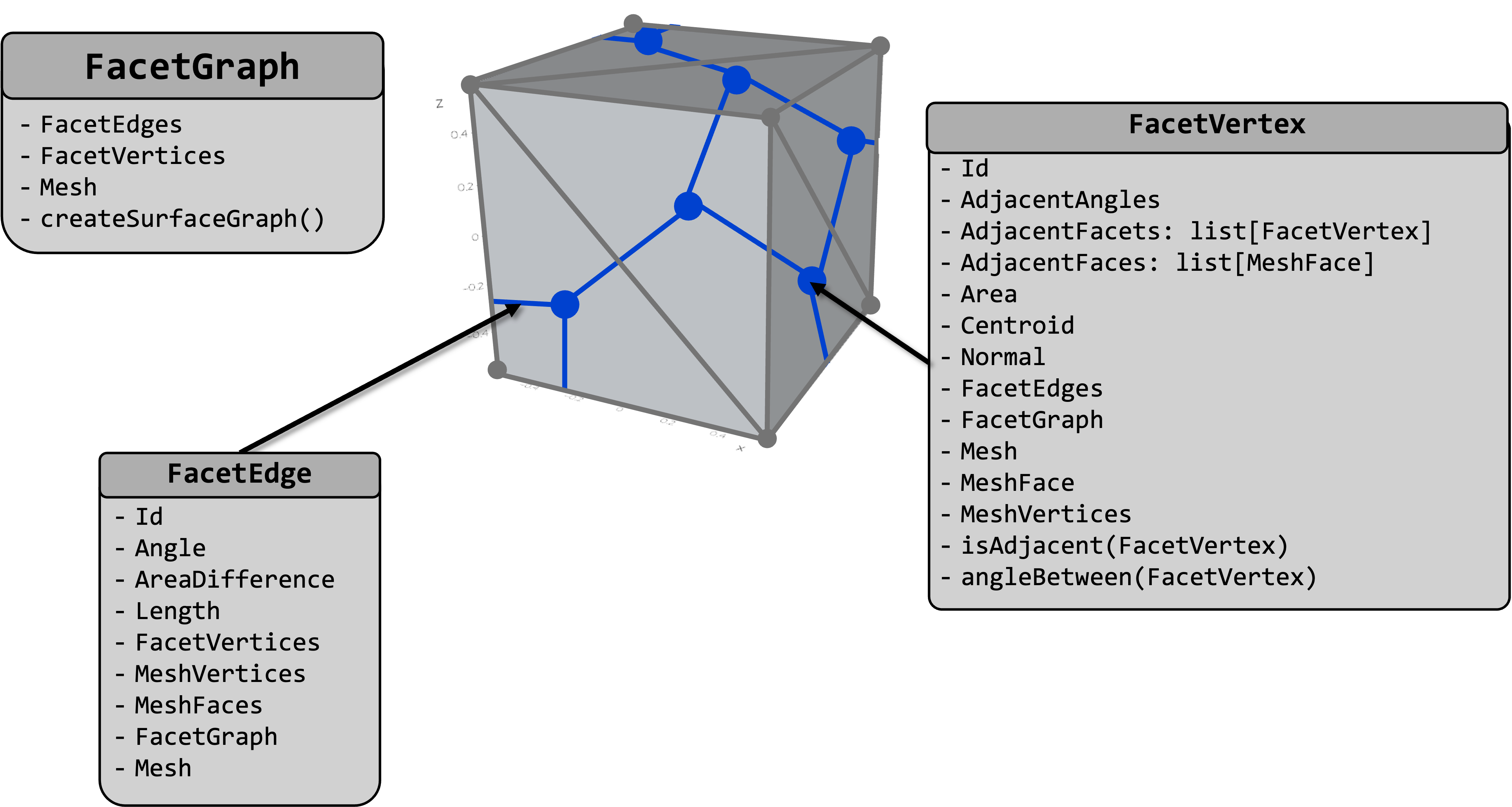}
    \caption{Description of data structures for our \texttt{FacetGraph} classes.}
    \label{fig:facetUML}
\end{figure}

The facet graph is represented by 3 classes (Figure \ref{fig:facetUML}): \texttt{FacetVertex}, \texttt{FacetEdge}, and \\\texttt{FacetGraph}. A \texttt{FacetVertex} represents a face on a 3D mesh, and includes  geometric properties from the \texttt{MeshFace} class. Additional information includes a list of adjacent \texttt{FacetVertex} objects, the angles between each \texttt{FacetVertex}, and list of \texttt{FacetEdge} objects the vertex is a part of.

\texttt{FacetEdge} inherits all information from \texttt{MeshEdge}, yet is retained at the facet level for consistency.

Finally, \texttt{FacetGraph} contains a list of \texttt{FacetVertex} and \texttt{FacetEdge} objects, and the \texttt{Mesh} that it was created from. The \texttt{FacetGraph} also includes a method to convert a \texttt{FacetGraph} to a \texttt{SurfaceGraph} with a user-configurable threshold number.
The conversion algorithm operates as described in Algorithm \ref{alg:partitionToSurface}.

\subsubsection{Implementation of Surface Graph}

\begin{figure}[h]
    \centering
    \includegraphics[width=\columnwidth]{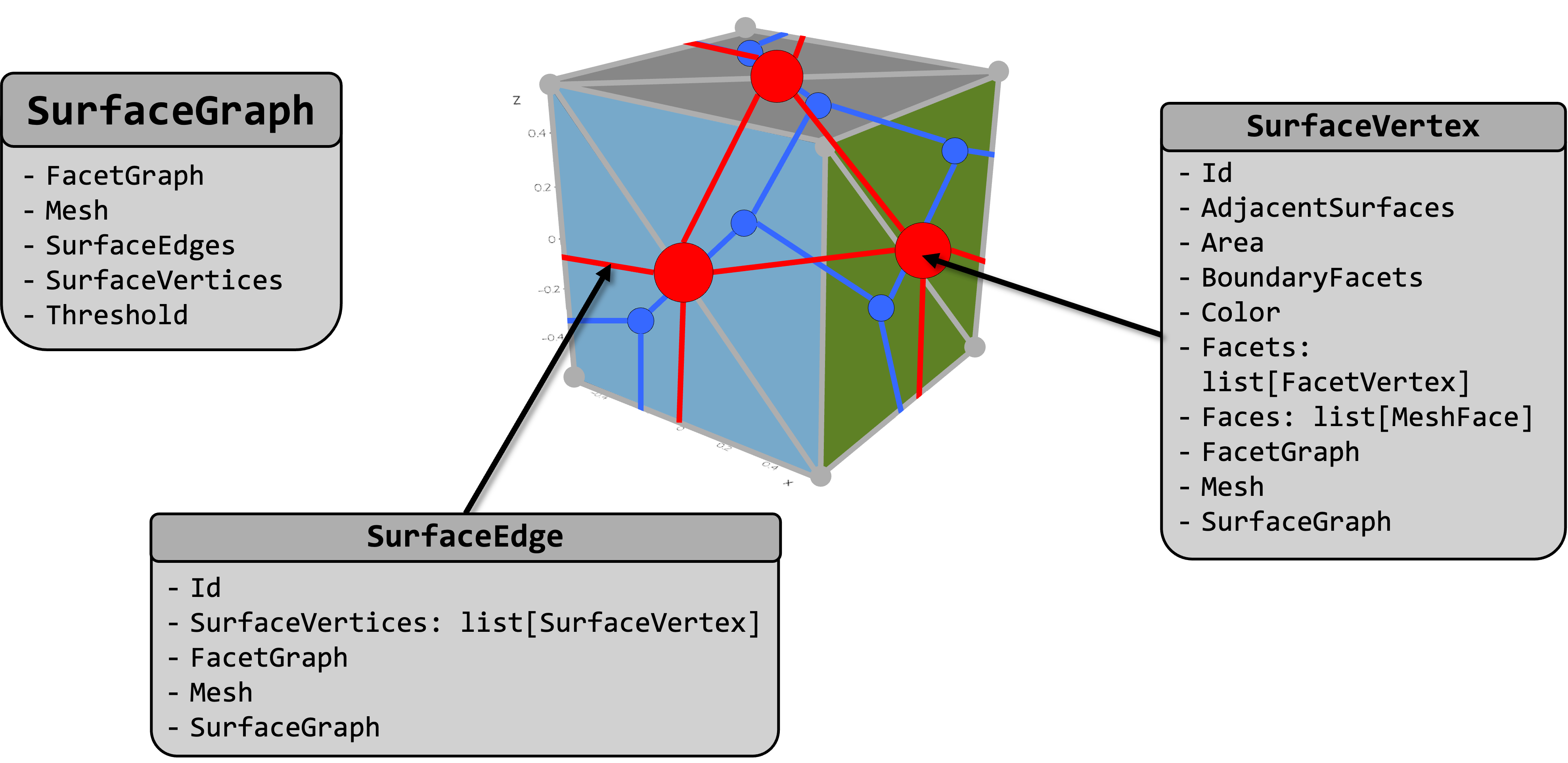}
    \caption{Description of data structures for our \texttt{SurfaceGraph} classes.}
    \label{fig:surfaceUML}
\end{figure}

The final layer of the structured graph we implement for now is the \texttt{SurfaceGraph} (Figure \ref{fig:surfaceUML}). Similar to the previous levels, this reduction contains objects representing a surface's \texttt{SurfaceVertex}, \texttt{SurfaceEdge}, and \texttt{SurfaceGraph}. 

The \texttt{SurfaceVertex} class contains a list of facets associated in the cluster. It also refers to the mesh faces, providing easy access to the points of the original mesh (i.e. the \texttt{MeshVertex}) if needed. Each surface vertex has a list of which facets form the boundary of the surface, computed as described in Algorithm \ref{alg:partitionToSurface}. The total area of all facets in the surface is maintained as well. 
\texttt{SurfaceEdge} lists which \texttt{SurfaceVertex} objects are adjacent. 

The \texttt{SurfaceGraph} has a list of \texttt{SurfaceVertex} and \texttt{SurfaceEdge} objects that make up the graph, and also contains references so the user can easily access the \texttt{FacetGraph} and \texttt{MeshGraph}.

Once the surface graph is computed, it can then be visualized as a color overlay on the facets of the model.

\subsection{Evaluation of SAMS}\label{sec:evaluation}

This section evaluates the performance of SAMS, a structured graph representation of STL files. We designed our experiment to assess the accuracy at which our package is able to populate mesh information. To evaluate our implementation, we assess the correctness of individual files across several mesh processing packages.

The overall research question which should be answered to evaluate SAMS is: \textit{Is our mesh representation accurate, compared with other state-of-the-art mesh packages?} 

SAMS is assessed for accuracy. To evaluate this, we load STL files with the package, record the mesh vertices and facets, and compare with widely used Python mesh processing packages to ensure that the information extracted from meshes by SAMS conforms to information extracted by those packages.

\subsubsection{Data Collection}
The package is evaluated using 2,779 STL files downloaded from Thingiverse \cite{thingiverse.com}. Each file on Thingiverse is provided with a form of Creative Commons License, which may restrict modification or distribution for profit. We use the selected files for evaluating the accuracy of our implementation.

\subsubsection{Experimental Objects}

We have collected 2,779 STL files, therefore we have one data set consisting of 2,779 experimental objects.

These objects vary in complexity. The largest file is 99.7 MB and the smallest is 684 bytes. The distribution of the number of facets, edges, and vertices in the data set is summarized in Figure \ref{fig:experimental-objects}.

\begin{figure}
    \centering
    \includegraphics[width=\columnwidth]{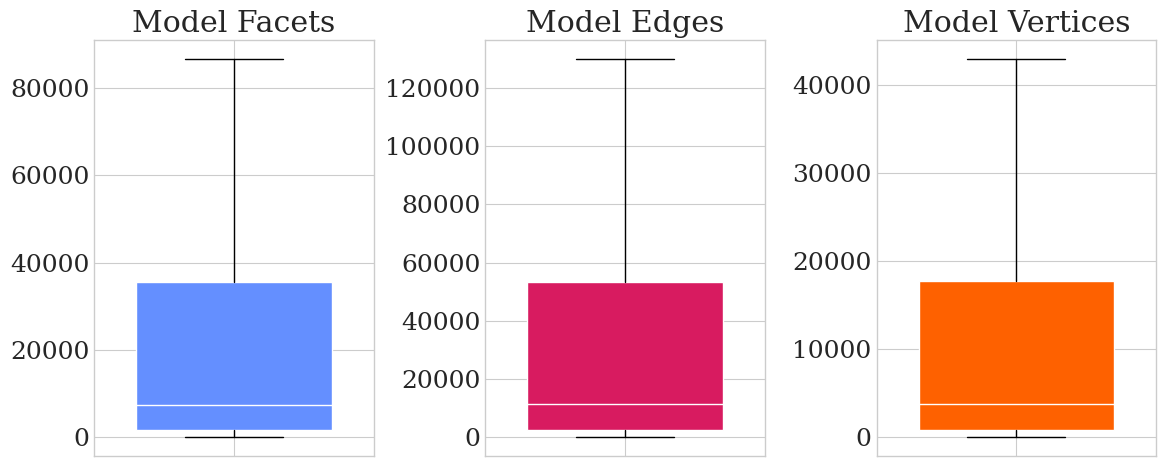}
    \caption{Boxplots summarizing the important statistics of our data set. Due to the size of the data set, outliers were not included. The mean number of facets is 77601.141 and the maximum is 1989076. The mean number of edges is 155200.87 and the maximum is 5967228. The mean number of vertices is 38780.614 and the maximum is 983040.}
    \label{fig:experimental-objects}
\end{figure}

\subsubsection{Variables and Metrics}
\textit{Dependent Variables.} 
The following metric is chosen to answer and evaluate the research question.

\textbf{Percentage match} is the percentage overlap of the respective number of facets, edges, and vertices contained in experimental objects across the independent variables.

\textit{Independent Variables.}
We use the dependent variables and test them across three different Python mesh processing software:

\begin{itemize}
    \item SAMS, the proposed implementation.
    \item {\sc Trimesh}, a Python package used in the Cura slicer \cite{dawson-haggerty2022}.
    \item {\sc PyMeshLab}, a mesh editing and processing software \cite{pymeshlab}.
\end{itemize}

By recording metrics across three different variables, this will allow us to successfully answer the research question. However, since no structured graph tools as we propose currently exist in this object-oriented formulation at all of the mesh, facet, and surface levels, we do not have baseline systems to compare to our work. 

\subsubsection{Study Operation}
The study was completed using a PC running Windows 11 Enterprise 25H2. The machine had 2 TB SSD and 32 GB DDR5 RAM. It tested the dependent variables using Python 3.12.0, trimesh 4.12.2, pymeshlab 2025.7.post1, and our implemented SAMS classes.

\subsubsection{Threats to Validity}

A primary threat to external validity is whether our experimental objects are representative. To mitigate that, we have downloaded over 2,000 experimental objects, ranging in file size and complexity. Some of these files are primitive structures, such as a long rectangular box, which has only 8 vertices and 12 faces. Others are very complex, with 993,157 vertices and 1,988,670 faces.
We believe the chosen objects sufficiently capture the range possible in STL complexity.

Regarding internal validity, the primary threat would be bias or errors in implementation. To mitigate this, the authors contributed to distinct components of the experiment. The 1st author designed the experiment and verified the results. 
The 2nd author implemented the SAMS graphs and executed the experiment. The 3rd author processed the experimental results.

\subsection{Results}\label{sec:results}
Now that SAMS has been evaluated experimentally against comparable mesh processing packages, the results of the experiments are now described in Section \ref{sec:results}. They will be discussed in Section \ref{sec:discussion}.

To evaluate the accuracy of SAMS, we loaded each file in our data set in each of the three evaluated mesh processing packages: SAMS, Trimesh, and PyMeshLab. For each file, we evaluated if the number of facets, vertices, and edges is the same across each package. The results are summarized in Table \ref{tab:acc-table}.

\begin{table}[]
\centering
\caption{Results of the research question of the SAMS' framework accuracy. Trimesh stores duplicated edges as opposed to de-duplicated storage by SAMS and PyMeshLab.}
\label{tab:acc-table}
\begin{tabular}{c|ccc|}
\cline{2-4}
 & \multicolumn{3}{c|}{Percentage Match With SAMS} \\ \hline
\multicolumn{1}{|c|}{\textbf{Package}} & \multicolumn{1}{c|}{\textbf{Edges}} & \multicolumn{1}{c|}{\textbf{Vertices}} & \textbf{Facets} \\ \hline
\multicolumn{1}{|c|}{Trimesh} & \multicolumn{1}{c|}{200\%} & \multicolumn{1}{c|}{100\%} & 100\% \\ \hline
\multicolumn{1}{|c|}{PyMeshLab} & \multicolumn{1}{c|}{100\%} & \multicolumn{1}{c|}{100\%} & 100\% \\ \hline
\end{tabular}%
\end{table}

SAMS matches 100\% with both Trimesh and PyMeshLab for faces and vertices. It matches 100\% with PyMeshLab for edges. When loading a model in Trimesh, by default it stores duplicated edges. However, we assessed that after de-duplicating, each of SAMS and Trimesh stores edges uniquely and correctly.

\section{Discussion}\label{sec:discussion}

The evaluation results are now discussed in detail. SAMS performs on par in terms of correctness. Here, we argue that SAMS is useful in terms of the knowledge stored in the classes, and the framework's ease of use for further high-level analysis.

\subsection{Selection of Experimental Objects}

The initial data collection began by gathering 2,780 experimental objects. 

One file was removed from the data set because of UTF8 filename errors on the machine used to operate the study. The removed file was named with Cyrillic characters, which caused compatibility issues in Windows.

Four ASCII files had syntax errors where the file ended with \texttt{end solid} instead of the standard \texttt{endsolid}. These were fixed by deleting the space. The files were otherwise unchanged and returned to the data set.

\subsection{Choice of Independent Variables}
There are other Python mesh processing packages which were not included as independent variables. These include Open3D, which was discarded because the edge calculation of the mesh needed to be done manually and thus would not have been an equal comparison. PyVista and Vedo were discarded because of errors opening files with VTK.

We restricted the evaluation to mesh-level accuracy because no publicly available software package provides functionality directly comparable to our full framework. In particular, our hierarchical, structured graph formulation is novel and therefore not directly benchmarkable against existing mesh processing libraries. Although comparisons at the mesh level are feasible, since we operate on identical mesh inputs and can assess outputs under equivalent conditions, there is currently no open-source baseline implementing an analogous graph representation framework. Establishing such a baseline would require implementing the proposed SAMS within the backends of multiple mesh processing packages to ensure a fair comparison. Given the substantial engineering effort required, we did not pursue this implementation in the present work.

\subsection{SAMS Accuracy}

SAMS correctly imports all relevant information from an STL file. With the exception of how Trimesh chooses to double the storage of mesh edges, there is no difference in correctness between SAMS and comparable mesh processing packages. Thus, SAMS is safe to use in mesh processing.

\subsection{Sources of Overhead}
In our loading of meshes, we noticed that SAMS was performing slower than Trimesh and PyMeshLab. We now describe several sources of overhead that may be contributing, as well as the benefits of the increased processing time in the following subsection.
A possible source of overhead in general is the processing time spent to parse ASCII instead of binary STL files. Binary files are much smaller and therefore take less time to process than the larger ASCII STL files.
The overall processing time and memory used can be assisted by multithreading, but we chose not to do this in order to represent how a typical end user could conduct this experiment.

Part of the benefit of SAMS can be found in the lack of dependencies: our proposed method is designed to have few dependencies. We do not use a graph library like NetworkX due to the known inefficiency of that package. 
Ultimately, the SAMS framework is roughly 10-15 times less efficient than the comparable packages. 
Trimesh has dependencies of NumPy, SciPy (for graphs) and Pyglet (for visualization purposes). PyMeshLab also depends on NumPy. Both of these packages have roughly linear processing times as well. Our implementation uses NumPy only in very limited circumstances to calculate the angle between two facets. This is because we found some inefficiencies with NumPy as we were implementing the construction of surface graphs.

\subsection{Information retained by SAMS}
We parsed a mesh from scratch using few dependencies because it allows us to store relevant information into our data structures. Since the SAMS framework supports the clustering of vertices, it can store useful information to provide during later analysis. Its modularity thus can be leveraged to account for higher levels of the structured graph. The facets can be clustered according to an algorithm such as Algorithm \ref{alg:partitionToSurface} to obtain a surface graph, which can then be used for analysis, such as detecting anomalous mesh vertices. 

In doing so, SAMS provides ease of information access. It is simple to transition between the mesh, facet graph, and surface graph levels. It is simple to query a facet's adjacent faces. It is simple to query which vertices are associated with a group of triangles. These queries are more complex in other packages, requiring a script or multiple chained commands.

Mesh processing packages can contain numerous useful functions, but SAMS does not include many of those. For example, we provide no visualization services, only the modular code structures. However, our package is designed to easily interface with other packages that provide visualization services, particularly PyMeshLab and Vedo. This enables us to keep overhead as light as possible, while reserving memory and compute for the necessary structures of the STL object.

Such an approach has success in program analysis. By processing information about the system or data upfront, more intense computations done later have drastic improvements in terms of efficiency and ease of use \cite{tsutano2019}. Therefore, we now consider the ways SAMS graphs may be utilized.

\section{Utility Scenarios and Threat Model}\label{sec:usecases}

Three scenarios of SAMS's utility are now considered. These include mesh operations (Section \ref{sec:re-tri}), material assurance (Section \ref{sec:material}), as well as the broader security of the additive manufacturing pipeline (Section \ref{sec:threatmodel}). 

\subsection{Scenario 1: Re-Triangulation of Mesh}\label{sec:re-tri}

De Floriani et al \cite{defloriani1989} previously used the structured graph, consisting of what we now call the mesh and facet graph, to re-triangulate an area of the mesh. The operation may be further enhanced with our introduction of the surface graph. For flat areas, a vertex in the surface graph maintains its boundary edges, making re-triangulation a simple process with an algorithm of the user's choosing. Curved areas may require manual inspection to ensure the local topology changes do not disturb the object's intended tolerances, but with a small enough threshold per surface, re-triangulation may be obtained similarly. By restricting focus as such to a specific area of the mesh, re-triangulation may be targeted to an area of the user's choice. Such an operation allows for reducing mesh file size as needed.

\subsection{Scenario 2: Material-Specific Assurance of Model Geometry}\label{sec:material}

Because we focus solely on geometry in the digital model file, our approach can generally assist users who wish to identify and quantify a specific geometric error before the part goes to be manufactured. 

The proposed approach can also be tailored to specific use cases. This is because additive manufacturing can be accomplished using a variety of materials. Each material has different vulnerabilities with regards to structural failure \cite{awdallah2025}. For example, PLA, a type of plastic, is susceptible to fracture at sharp corners. Using our facet graph, we can compute the surfaces using Algorithm \ref{alg:partitionToSurface}, and set a high angle threshold such as $90^{\circ}$ for clustering. Doing so reveals areas of the part that could pose an issue if PLA were used for manufacturing.

Corners also present particular geometric issues in 3D printing of concrete \cite{bos2016}. It is recommended that sharp corners be rounded to avoid differences in the rate of material deposition. The exact threshold needed depends on the machine and type of filament used. Cantilevering geometry also presents an issue. SAMS's surface graph can detect sharp corners as well as transitions to sharply descending geometry.

Parts can also be made with metal such as steel, titanium, copper, and aluminum \cite{markforged-metal}. Metal components can be strong enough to be used in other machinery. Small changes to part geometry can affect tolerances in interlocking parts. For example, in a common part used to test 3D printers, a gear bearing with slightly modified coordinates of existing vertices can result in the components locking, as well as uneven wear and tear. Small mutations in vertex coordinates can currently be identified by visual inspection. However, SAMS' surface graph can make the process of detecting such changes easier. Currently, if the modification is large enough and the threshold is small enough, an instance of vertex shifting may be detected with SAMS by computationally identifying mesh vertices that are a part of many small surfaces (Figure \ref{fig:vshift}). We are actively working to develop an automated detection method of this modification that has a high throughput for many kinds of models.

\begin{figure}[h]
    \centering
    \includegraphics[width=\linewidth]{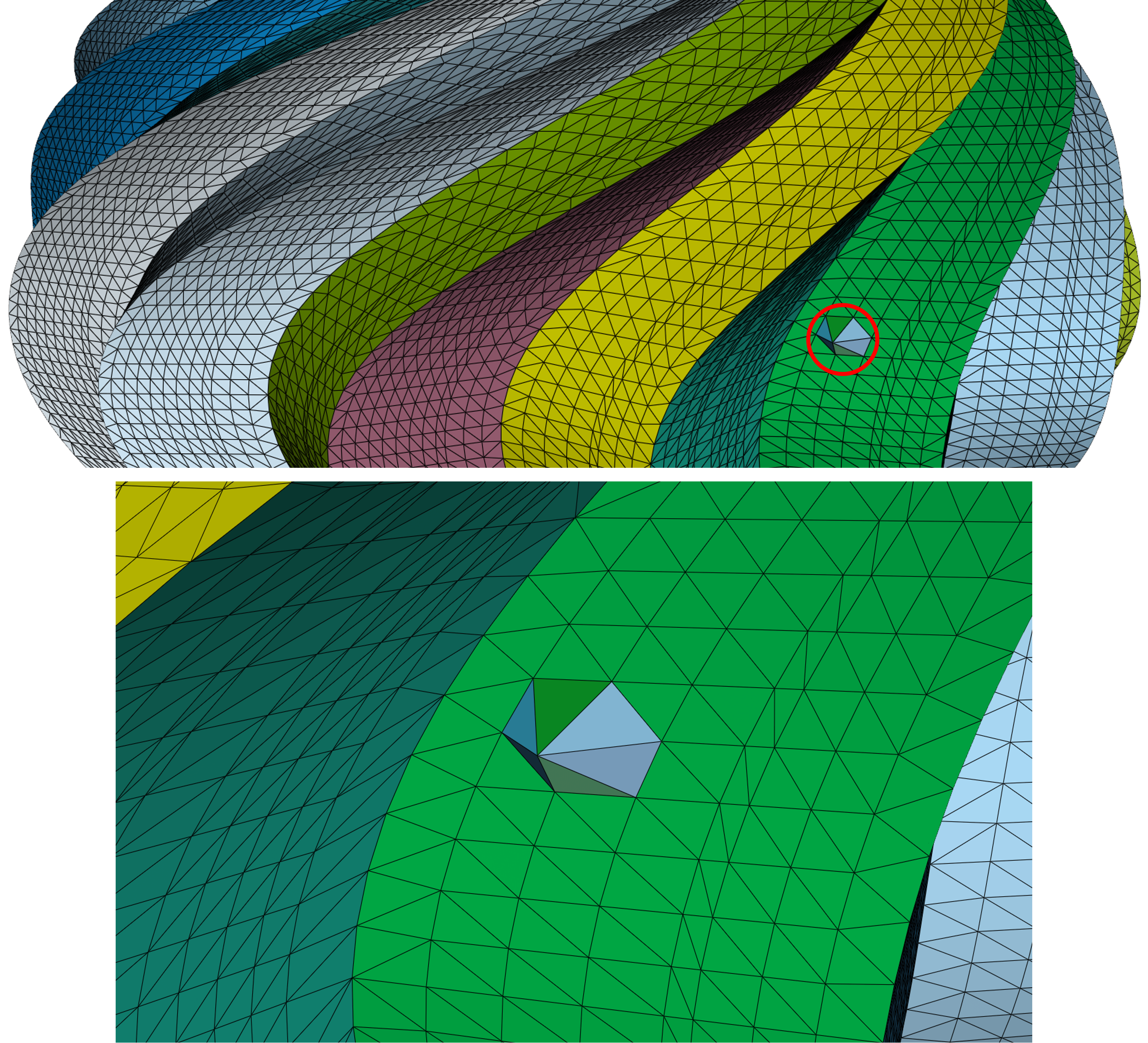}
    \caption{A shifted vertex which can be identified using SAMS graphs as the intersection of several small surfaces.}
     \Description[A vase with a spike]{A figure with two parts: The top shows a vase with colored with multicolored bands. The vase has a small spike on the right hand side which is emphasized with a red circle. The bottom is zoomed in on the spike. It is defined by multiple differently colored triangles.}
    \label{fig:vshift}
\end{figure}

\subsection{Scenario 3: Security of Additive Manufacturing Pipeline}\label{sec:threatmodel}
We seek to use the SAMS software to identify errors, defects, or irregularities in STL files. The research reported here is useful for graphics, the additive manufacturing (AM) supply chain, and computer security.
A survey on attack taxonomies to the AM supply chain has been accomplished by Kumar et al \cite{kumar2025}. It identifies a number of possible threats facing AM. It expands on the work reported by \cite{yampolskiy2018}. We provide a threat model based on that taxonomy of attack goals, targets, and methods.

Our work focuses on the stage between conversion of a CAD model to STL and generation of tool-path code. We focus on the ``Data Breach and Manipulation" attack goal. This is tied to the attack target ``STL file integrity". The target involves the following attack methods:
\begin{itemize}
    \item 
"Reverse engineering of AM design", 
\item "Cyber intrusion on AM equipment", and 
\item "Exploitation of supply chain weakness" attack methods. 
\end{itemize}
Our framework aims to combat attack philosophies matching one or more of these methods.
We aim to combat attacks which can specifically sabotage STL mesh files, since the effects of intrusion or manipulation of a file can propagate to later stages of manufacturing, leading to wasted labor and material.

Current work in this area includes detecting of voids, or negative space, in 3D meshes \cite{roscoe2025}. A future endeavor is to detect vertex shifting as described in the previous subsection. STL files that have been modified with a vertex shift operation could be interpreted as just introducing noise to a model, but the modification has significant implications for the security of the industry.
Current and likely future legislation in the United States has involved the attempt to ban printing of certain parts, usually those involved in the manufacturing of weapons that are created in whole or in part by additive manufacturing \cite{nybills9005c,ab2047}. The legislation involves blueprint detection (in the case of California's AB 2047 \cite{ab2047}) of 3D models, a method that would need to verify that the model being printed does not match that of a model known to be a threat. Though not currently specified at the time of writing, a detection method could involve computer vision, AI processing of a model, or checking a model's hash against a predetermined database. However, depending on how the requirement is implemented, a simple vertex shift could assist individuals in evading a hash check of a given model. Multiple vertex shifts could lead to AI or computer vision failing to flag the model as suspicious. Therefore, vertex shifting of models is an important area of research in additive manufacturing to investigate. Some work has been progressed in this area in \cite{zhao2021}. SAMS's initial investigation of this is also a promising first step.

Now the related work to our proposed method is presented.

\section{Related Work}\label{sec:related}

In this paper we proposed a structured graph framework for 3D meshes for the purpose of bolstering the security of additive manufacturing. This framework is hierarchical, with the aim of identifying geometric issues in a mesh and in more complex mesh analysis operations. We now address similar or related works to ours of two kinds: those which specifically use a hierarchical graph to analyze or process meshes, and those involving graph modeling of meshes in some way.

First, STL-encoded triangle meshes are routinely analyzed and manipulated using mesh processing toolchains. For computational efficiency, mesh topology is often mapped to a relational graph representation enabling graph algorithms to be applied directly.

Several software frameworks explicitly support mesh-to-graph representations to support geometric computations, such as the Trimesh Python library \cite{dawson-haggerty2022}. Graph representations are also widely used as inputs to convolutional and graph neural network models for tasks such as shape classification \cite{song2020,tang2021,li2023} and mesh segmentation \cite{fotopoulou2020,vecchio2023}.

Beyond flat adjacency graphs, hierarchical graph structures have been proposed to encode higher-level part semantics. For example, {\sc StructureNet} \cite{mo2019a} constructs a part hierarchy and propagates information through an encoder-decoder architecture, yielding a graph-structured representation of semantic features. Our approach similarly employs a hierarchy, but additionally supports information exchange across multiple hierarchy levels as required by a given operation; we implement this via structured graphs.

Dual graphs of triangulated networks were abstracted into structured graphs by De Floriani et al. \cite{defloriani1989}. While we build on the general notion of a structured graph, our objectives and graph representation differ: their abstraction targets a 2D plane embedding to support hierarchical triangulation, whereas we preserve the 3D mesh information and treat the hierarchy as a reusable representation for analyses beyond triangulation, including additive manufacturing-specific use cases.

Finally, graph-based abstractions have been used for security-oriented modeling in additive manufacturing; Rahman et al. use an attack graph model for risk analysis \cite{rahman2024}. In contrast, we use graphs to abstract the STL geometry/topology itself, with the goal of deriving or assessing security-relevant properties from the part representation.

Now we address works which use some form of graph modeling to assist in mesh computations or analysis. 

Facet graphs are used as input to graph convolutional neural networks in \cite{tang2021} and \cite{lei2023} for shape classification. The studies \cite{pfaff2020, lee2019, ranjan2020,keuper2015} assist with various mesh processing applications such as encoding/decoding \cite{pfaff2020}, studying connectivity \cite{lee2019,ranjan2020}, and image decomposition \cite{keuper2015}. A facet graph is used to help reconstruct a mesh volumetrically in \cite{hornung2006}. 
Studies which use a mesh graph data structure include \cite{fotopoulou2020,vecchio2023}. These studies use the graph for mesh segmentation in the case of \cite{fotopoulou2020} to schedule complex operations in computer vision.

By sampling the facet graph iteratively and re-drawing adjacencies at each step, a series of hierarchical graphs are constructed by \cite{li2023}. The final graph is used as input to a neural network that is robust for mesh errors and can be applied to shape classification. Our proposed work has different strata of graphs. The cited work samples vertices. Our approach includes a heuristic for clustering according to geometric features. Furthermore, we treat those clusters as an entirely separate level of graph. 

Hex meshes can use dual sheets to reduce mutual intersections \cite{takayama2019}. THe notion of dual is used to accomplish intersection reduction, as well as preserve sharp features when hexahedralizing a mesh. The cited work is similar to our surface graph which can highlight sharp details with a user-set threshold, though the surface graph is part of a hierarchy of graphs and currently works only with triangle meshes, not with hex meshes.

Primal and dual graphs are constructed in \cite{milano2020}. Their definition of dual is different: their primal graph is analogous to our facet (dual) graph, and their dual graph is constructed by associating each edge of the primal graph (i.e. every triangle) with a vertex). In this way, the dual graph is constructed by the edges of the original mesh, and two vertices of the dual graph are joined if they compose the same triangle.
Both the dual and primal graphs in their topological construction are equivalent to our facet graph. The cited work also encodes mesh information such as dihedral and internal angles to each node of the graph. We also encode dihedral angles to the facet graph vertices. Our method proposes clustering of graph vertices for further analysis.

A mesh is split into convex surfaces by \cite{mamou2009}. The dual graph of the mesh is used to collapse half-edges, and a concavity measure is used. Our method supports these operations as well, and we encode more than just concavity into our dual graph vertices.

Denoising in light of 3D surface scanners is accomplished in \cite{zhao2021}. The method relies on local area and does not split the area on sharp edges, as the study's aim is to reduce the noise in a mesh. We recommend this paper to those seeking an alternative reduction of the facet graph.

In the area of security, both \cite{sturm2017} and \cite{belikovetsky2017} address the potential for model defects to impede a part's function and stability. This initial potential is confirmed in attack taxonomy surveys conducted by \cite{yampolskiy2018} and \cite{kumar2025}. The facet graph is also used via its connected components to identify shells in \cite{roscoe2025}, which are then analyzed for self-intersection to detect void defects in STL files. 

While the cited studies may use elements of graph modeling, their use is so far limited to the particular application, such as shape classification, mesh segmentation or processing, and denoising. To the best of our knowledge, SAMS links these applications together and provides a comprehensive graph model and framework. We address its graph-theoretic properties and provide an implementation for the unified graph representation.

\section{Conclusion}\label{sec:conclusion}

This paper has described a structured graph framework for the Security of Additive Manufacturing Systems. The framework, SAMS, is structured, according to levels of abstraction of a triangle mesh 3D model. We formally define this framework, address its graph-theoretic properties, and identify several scenarios for its practical use. While the processing time of our implementation is higher than comparable mesh processing packages, we establish that our graph classes are useful in their modular, adaptable structure, to enable more complex computations and analysis to be done efficiently at later stages. The proposed framework is well-suited to be used for mesh analysis and ensuring dependability of 3D model files.

To facilitate faster load times and computation, we plan to update the implementation of SAMS and use pybind11 to implement the core functionality in C++ in order to reduce overhead and improve performance. We also plan to expand the compatibility of SAMS to support other polygonal mesh types such as 3MF and OBJ.

\begin{acks}
    The authors wish to thank Viet Ninh for compiling the data set used in the experiments.
\end{acks}
\bibliographystyle{acm}
\bibliography{STLgraph.bib}

\end{document}